\documentclass[11pt,letterpaper]{article}
\usepackage[margin=1in]{geometry}
\usepackage{fullpage}
\usepackage{amsmath, amssymb,graphicx}
\usepackage{amsthm}
\usepackage{color}
\usepackage{verbatim, xspace}
\usepackage{algorithm}
\usepackage{algpseudocode}
\usepackage{thm-restate}
\usepackage{enumerate} 
\usepackage{subcaption}

\graphicspath{ {figures/} }
\newtheorem{theorem}{Theorem}[section]
\newtheorem{definition}[theorem]{Definition}
\newtheorem{lemma}[theorem]{Lemma}

\newtheorem{fact}[theorem]{Fact}

\newenvironment{proof-sketch}{\noindent{\bf Sketch of Proof}\hspace*{1em}}{\qed\bigskip}
\newenvironment{proof-idea}{\noindent{\bf Proof Idea}\hspace*{1em}}{\qed\bigskip}
\newenvironment{proof-of-lemma}[1]{\noindent{\bf Proof of Lemma #1}\hspace*{1em}}{\qed\bigskip}
\newenvironment{proof-attempt}{\noindent{\bf Proof Attempt}\hspace*{1em}}{\qed\bigskip}

\newcommand{\so}{\textsc{Steiner Orientation}\xspace}
\newcommand{\dmc}{\textsc{Directed Multicut}\xspace}
\newcommand{\mmax}{\textsc{Max}\xspace}

\newcommand{\eps}{\varepsilon}
\newcommand{\ex}{\mathbb{E}}
\newcommand{\Oh}{\mathcal{O}}

\newcommand{\og}{\tilde{G}}

\newcommand{\problembox}[4]{
  \begin{center}
    \fbox{
      \parbox{0.85\columnwidth}{
        #1\\[0.3em]
        \renewcommand{\tabcolsep}{3pt}
        \begin{tabular}{rp{0.7\columnwidth}}
          \textit{Input:\ } & #2\\
          \textit{Parameter:\ } & #3\\
          \textit{Task:\ } & #4
        \end{tabular}
      }
    }
  \end{center}
}

\usepackage{color}
\def\DEBUG{true}
\ifdefined\DEBUG
  \newcommand{\mic}[1]{\textcolor{red}{#1}}
  \newcommand{\rah}[1]{{\color{blue}{#1}}}
  \newcommand{\micr}[1]{\marginpar{\small \textcolor{red}{$\bullet$ #1}}}
  \newcommand{\rahr}[1]{\rem{\textcolor{blue}{$\bullet$ #1}}}
\else
  \newcommand{\mic}[1]{}
  \newcommand{\rah}[1]{}
  \newcommand{\micr}[1]{}
  \newcommand{\rahr}[1]{}
\fi

\title{Parameterized inapproximability for Steiner Orientation by~Gap~Amplification}
\author{Micha{\l} W{\l}odarczyk\footnote{Address: \texttt{m.wlodarczyk@tue.nl}. The main part of the work has been done when the author was a Ph.D. student at the University of Warsaw and it was a part of the project TOTAL, that has received funding from the European Research Council (ERC) under the European Union's Horizon 2020 research and innovation programme (grant agreement No 677651).} \\ Eindhoven University of Technology}
\date{}

\begin{document}

\maketitle
\begin{abstract}
    In the $k$-\so problem, we are given a~mixed graph, that is, with both directed and undirected edges, and a~set of $k$ terminal pairs.
The goal is to find an~orientation of the~undirected edges that maximizes the number of terminal pairs for which there is a~path from the~source to~the~sink. 
The problem is known to be W[1]-hard when parameterized by $k$ and hard to approximate up to some constant for FPT algorithms assuming Gap-ETH.
On the other hand, no~approximation factor better than~$\Oh(k)$ is known.

We show that $k$-\so is unlikely to admit an~approximation algorithm with any constant factor, even within FPT running time. 
To obtain this result, we construct a~self-reduction via~a~hashing-based gap amplification technique, which turns out useful even outside of the FPT paradigm.
Precisely, we~rule~out any~approximation factor of~the form $(\log k)^{o(1)}$ for~FPT algorithms (assuming FPT $\ne$ W[1]) and $(\log n)^{o(1)}$ for~purely polynomial-time algorithms (assuming that the~class~W[1] does~not admit randomized FPT algorithms).
This constitutes a~novel inapproximability result for~polynomial-time algorithms obtained via~tools from the FPT theory.
Moreover, we~prove $k$-\so to~belong to~W[1], which entails W[1]-completeness of $(\log k)^{o(1)}$-approximation for~$k$-\so.
This provides an~example of a~natural approximation task that is complete in~a~parameterized complexity class.

Finally, we apply our technique to the maximization version of directed multicut -- \textsc{Max $(k,p)$-Directed Multicut} -- where we are given a~directed graph, $k$ terminals pairs, and a~budget $p$.
The goal is to maximize the number of separated terminal pairs by removing $p$ edges.
We present a~simple proof that the problem admits no FPT approximation with factor $\Oh(k^{\frac 1 2 - \eps})$ (assuming FPT $\ne$ W[1]) and no polynomial-time approximation with ratio $\Oh(|E(G)|^{\frac 1 2 - \eps})$ (assuming  NP $\not\subseteq$ co-RP).
\end{abstract}

\newpage

\section{Introduction}

In the recent years new research directions emerged in the intersection of the two theories aimed at tackling NP-hard problem:
parameterized complexity and approximation algorithms.
This led to numerous results combining techniques from both toolboxes.
The main goal in this area is to obtain an algorithm running in time $f(k)\cdot |I|^{\Oh(1)}$ for an~instance $I$ with parameter $k$, that finds a solution of value not worse than $\alpha$ (the approximation factor) times the value of the optimal solution.
They are particularly interesting for problems that are both W[1]-hard and at the same time cannot be well approximated in polynomial time~\cite{esa, bsn, kmedian, kcut, doct}.
On the other hand, some problems remain resistant to approximation even in this paradigm.


Obtaining polynomial-time approximation lower bounds under the assumption of P $\ne$ NP is challenging, because it usually requires
to prove NP-hardness of a~\emph{gap problem}.
This is a problem, where one only needs to distinguish instances with the value of optimal solution at least $C_1$ from those with this value at most $C_2$. 
This provides an~argument that one cannot obtain any approximation factor better than the gap, i.e., $\frac{C_1}{C_2}$, as long as P $\ne$ NP.

A~road to such lower bounds has been paved by the celebrated PCP theorem~\cite{pcp-arora}, which gives an alternative characterization of the class NP.
The original complicated proof has been simplified by Dinur~\cite{pcp} via the technique of \emph{gap amplification}:
an~iterated reduction from a gap problem with a small gap to one with a larger gap.
When the number of iterations depends on the input size, this allows us to start the chain of reductions from a problem with no constant gap.
However, this is only possible when we can guarantee that the size of all created instances does not grow super-polynomially.

The process of showing approximation lower bounds becomes easier with an additional assumption of the \emph{Unique Games Conjecture}~\cite{ugc}, which states
that a~particular gap version of the \textsc{Unique Games} problem is NP-hard.
This makes it possible to start a~reduction from a~problem with an~already relatively large gap.
The reductions based on {Unique Games Conjecture} provided numerous tight approximation lower bounds~\cite{ugc2, ugc3, ugc1}.

A~parameterized counterpart of the hardness assumption P $\ne$ NP is FPT $\ne$ W[1],
which is equivalent to the statement that \textsc{$k$-Clique} $\not\in$ FPT, that is, \textsc{$k$-Clique}\footnote{We attach the parameter to the problem name when we refer to a parameterized problem.} does not admit an~algorithm with running time of the form $f(k)\cdot |I|^{\Oh(1)}$.
Analogously to the classical complexity theory, proving hardness of an~approximate task relying only on FPT $\ne$ W[1] is difficult but possible.
A~recent result stating that the gap version of \textsc{$k$-Dominating Set} is W[1]-hard
(for the gap being any computable function $F(k)$) required gap amplification through a~distributed PCP theorem~\cite{dominating-set}.

Again, the task becomes easier when working with a~stronger hardness assumption: \emph{Gap Exponential Time Hypothesis}\footnote{Gap-ETH is a stronger version of the \emph{Exponential Time Hypothesis} (ETH), according to which one requires exponential time to solve \textsc{3-CNF-SAT} \cite{eth}.} (Gap-ETH) states that there exists $\eps > 0$ so that one requires exponential time to distinguish satisfiable \textsc{3-CNF-SAT} instances from those where only a~fraction of~$(1-\eps)$ clauses can be satisfied at once~\cite{gapeth2, gapeth-1}.
Gap-ETH is a~stronger assumption than FPT $\ne$ W[1], i.e., the latter implies the first, and it sometimes turns out more convenient since it already provides hardness for a problem with a gap. 
There are many recent examples of using Gap-ETH for showing hardness of parameterized  approximation~\cite{even-set, gapeth, bsn, planar-hardness, kmedian, doct}.

Our contribution is a novel gap amplification technique which exploits the fact that in a~parameterized reduction we can afford an~exponential blow-up with respect to the parameter.
It~circumvents the obstacles related to PCP protocols and, together with a~hashing-based technique, allows us to construct relatively simple self-reductions for problems on directed graphs.

\paragraph{Steiner Orientation}
In the $k$-\so problem we are given a mixed graph, that is, with both directed and undirected edges, and a set of $k$ terminal pairs.
The goal is to find an orientation of the undirected edges that maximizes the number of terminal pairs for which there is a path from the source to the sink. 

Some of the first studies on the $k$-\so problem (also referred to as the \textsc{Maximum Graph Orientation} problem) were motivated by modeling protein-protein interactions (PPI)~\cite{logn}
and protein-DNA interactions (PDI)~\cite{hal,logn-loglogn}.
Whereas PPIs interactions could be represented with undirected graphs, PDIs required introducing mixed graphs.
Arkin and Hassin~\cite{nphard} showed the problem to be NP-hard, but
polynomially solvable for $k=2$.
This result was generalized by Cygan et al.~\cite{xp}, who presented an $n^{\Oh(k)}$-time algorithm, which implied that the problem belongs to the class XP when parameterized by $k$ (cf. \cite{tapas} for different choices of parameterization).

The $k$-\so problem
has been proved to be W[1]-hard (see Section~\ref{sec:prelim}) by Pilipczuk and Wahlstr\"{o}m~\cite{pilipczuk-multicut}, which
makes it unlikely to be solvable in time \(f(k)\cdot |I|^{\Oh(1)}\).
The W[1]-hardness proof has been later strengthened to work on planar graphs and to give a~stronger running time lower bound based on 
ETH~\cite{planar-hardness}, which is essentially tight with respect to the $n^{\Oh(k)}$-time algorithm.

The approximation of \so has been mostly studied on~undirected graphs, where the problem reduces to optimization over trees by contracting \mbox{2-connected} components~\cite{xp}.
Medvedovsky et al.~\cite{logn} presented an $\Oh(\log n)$-approximation
and actually proved that this bound is achievable with respect to the total number of terminals $k$.
The approximation factor has been improved to $\Oh(\log n / \log\log n)$~\cite{logn-loglogn} and later to $\Oh(\log k / \log\log k)$~\cite{xp} by observing that one can compress an undirected instance to a tree of size $\Oh(k)$.
A lower bound of $\frac{12}{11}-\eps$ (based on P $\ne$ NP) has been obtained via a~reduction from \textsc{Max Directed Cut}~\cite{logn}.
Medvedovsky et  al.~\cite{logn} posed a question of tackling the maximization problem on mixed graphs, which was partially addressed by
Gamzu et al.~\cite{logn-loglogn} who provided a polylogarithmic approximation in the case where the number of undirected components on each source-sink path is bounded by a~constant.

The decision problem whether all the terminal pairs can be satisfied is polynomially solvable when restricting input graphs to be undirected
\cite{polytime}, which makes the maximization version fixed-parameter tractable, by simply enumerating all subsets of terminals.
The maximization version on mixed graphs is far less understood from the FPT perspective.
It is unlikely to be exactly solvable since the decision problem is W[1]-hard, but can we approximate it within a~reasonable factor?
The reduction by Chitnis et al.~\cite{planar-hardness} implies that, assuming Gap-ETH, $k$-\so cannot be approximated within factor $\frac {20} {19} - \eps$ on mixed graphs, in running time $f(k)\cdot n^{\Oh(1)}$.
Using new techniques introduced in this paper, we are able to provide stronger lower bounds based on a weaker assumption.

\paragraph{Related work}
\label{sec:recent}
Some examples of the new advancements in parameterized approximations are 1.81-approximation for \textsc{$k$-Cut}~\cite{kcut} (recently improved to $(\frac 5 3 + \eps)$~\cite{kcut-improved}) or $(1+\frac{2}{e}+\eps)$-approximation for \textsc{$k$-Median}~\cite{kmedian}, all running in time \(f(k)\cdot n^{\Oh(1)}\).
The first result beats the factor 2 that is believed to be optimal within polynomial running time and the latter one reaches the polynomial-time lower bound, which is a long standing open problem for polynomial algorithms.
For \textsc{Capacitated $k$-Median}, a constant factor FPT approximation has been obtained~\cite{esa}, whereas the best-known polynomial-time approximation factor is $\Oh(\log k)$.
Another example is an~FPT approximation scheme for the planar case of \textsc{Bidirected Steiner Network}~\cite{bsn}, which does not admit a~polynomial-time approximation scheme unless P = NP.

On the other hand several problems have proven resistant to such improvements.
Chalermsook et al.~\cite{gapeth} showed that under the assumption of Gap-ETH there can be no parameterized approximations with ratio $o(k)$ for \mbox{\textsc{$k$-Clique}} or~\textsc{$k$-Biclique} and none with ratio
$f(k)$ for \textsc{$k$-Dominating Set} (for any computable function $f$).
They have also ruled out $k^{o(1)}$-approximation for \textsc{Densest $k$-Subgraph}.
The cited FPT approximation for \textsc{$k$-Median} has a~tight approximation factor assuming Gap-ETH~\cite{kmedian}.

Subsequently, efforts have been undertaken to weaken the complexity assumptions on which the lower bounds are based.
For the \textsc{$k$-Dominating Set} problem
{Gap-ETH} has been replaced with a more established 
hardness assumption that FPT $\ne$ W[1]~\cite{dominating-set}.
Marx~\cite{marx-circuit} has proven parameterized inapproximability of \textsc{Monotone $k$-Circuit SAT} under the even weaker assumption that FPT $\ne$ W[P].
Lokshtanov et al.~\cite{doct} introduced the \emph{Parameterized Inapproximability Hypothesis} (PIH), that is weaker than Gap-ETH and stronger than FPT $\ne$ W[1],
and used it to rule out an~FPT approximation scheme for \textsc{Directed $k$-Odd Cycle Transversal}.
PIH turned out to be a~sufficient assumption to argue there can be no FPT algorithm for $k$-\textsc{Even Set}~\cite{even-set}.

\section{Overview of the results}
Our main inapproximability result is a W[1]-hardness proof for 
the gap version of $k$-\so with the gap $q = (\log k)^{o(1)}$.
This means that the problem is unlikely to admit~an algorithm with running time $f(k)\cdot |I|^{\Oh(1)}$.

\begin{restatable}{theorem}{mainso}
\label{thm:mainso}
Consider a function $\alpha(k) = (\log k)^{\beta(k)}$, where $\beta(k) \rightarrow 0$ is computable and non-increasing.
It is \emph{W[1]}-hard to distinguish whether for a~given instance of $k$-\so:
\begin{enumerate}
\itemsep0em
    \item there exists an orientation satisfying all $k$ terminal pairs, or
    \item for all orientations the number of satisfied pairs is at most $\frac{1}{\alpha(k)}\cdot k$.
\end{enumerate}
\end{restatable}

The previously known approximation lower bound for FPT algorithms, $\frac{20}{19}-\eps$, was obtained via a linear reduction from $k$-\textsc{Clique} and was based on Gap-ETH~\cite{planar-hardness}.
Our reduction not only raises the inapproximability bar significantly, but also weakens the hardness assumption (although we are not able to enforce the planarity of the produced instances, as in~\cite{planar-hardness}).
In fact, we begin with the decision version of $k$-\so and introduce a gap inside the self-reduction.
What is interesting, we rely on totally different properties of the problem than in the W[1]-hardness proof~\cite{pilipczuk-multicut}: that one required gadgets with long undirected paths and we introduce only new directed edges.


This result is also interesting from the perspective of the classical (non-parameterized) approximation theory.
The best approximation lower bound known so far has been $\frac{12}{11} - \eps$~\cite{logn},
valid also for undirected graphs.
Therefore we provide a new inapproximability result for polynomial algorithms, which is based on the assumption that FPT $\ne$ W[1].
Moreover, restricting to a~purely polynomial running time allows us to rule out approximation factor depending on $n$ (rather than on~$k$) with a~slightly stronger assumption, which is required because the reduction is randomized (see Section~\ref{sec:prelim} for the formal definition of~a~false-biased FPT algorithm).


\begin{restatable}{theorem}{polyso}
Consider a function $\alpha(n) = (\log n)^{\beta(n)}$, where $\beta(n) \rightarrow 0$ is computable and non-increasing.
Unless the class W[1] admits false-biased FPT algorithms, there is no polynomial-time algorithm that, given an instance of \so with $n$ vertices and $k$ terminal pairs, distinguishes between the following cases:
\begin{enumerate}
\itemsep0em
    \item there exists an orientation satisfying all $k$ terminal pairs, or
    \item for all orientations the number of satisfied pairs is at most $\frac{1}{\alpha(n)}\cdot k$.
\end{enumerate}
\end{restatable}

A~similar phenomenon, that is, novel polynomial-time hardness based on an~assumption from parameterized complexity, has appeared in~the~work on \textsc{Monotone $k$-Circuit SAT}~\cite{marx-circuit}.
Another example of this kind is  polynomial-time approximation hardness for \textsc{Densest $k$-Subgraph} based on ETH~\cite{eth-dks}.

\paragraph{W[1]-completeness}
So far, the decision version of $k$-\so  has only been known to be W[1]-hard~\cite{pilipczuk-multicut} and to belong to XP~\cite{xp}.
We establish its exact location in the W-hierarchy.
A~crucial new insight is that we can assume the solution to be composed of~$f(k)$ pieces, for which we only need to check if they match each other, and this task reduces to $k$-\textsc{Clique}.

\begin{restatable}{theorem}{mainhard}
$k$-\so is \emph{W[1]}-complete.
\end{restatable}

We hereby solve an open problem posted by Chitnis et al.~\cite{planar-hardness}.
What is more, this implies that $(\log k)^{o(1)}$-\textsc{Gap} $k$-\so belongs to W[1] (see Section~\ref{sec:prelim} for formal definitions).
Together with Theorem~\ref{thm:mainso} this entails W[1]-completeness.
Another gap problem with this property is \textsc{Maximum $k$-Subset Intersection}, introduced for the purpose of proving W[1]-hardness of $k$-\textsc{Biclique}~\cite{biclique}.
We are not aware of~any other natural gap problem being complete in a~parameterized complexity class.

\paragraph{Directed Multicut}
As another application of our technique, we present a~simple hardness result for the gap version of \textsc{Max $(k,p)$-Directed Multicut} with the gap $q=k^{\frac 1 2 - \eps}$.
We show that even if we parameterize the problem with both the number of terminal pairs $k$ and the size of the cutset $p$, then we essentially cannot obtain any approximation ratio better than $\sqrt{k}$.

\begin{restatable}{theorem}{maindmc}
\label{thm:dmc-w1}
For any $\eps > 0$ and function $\alpha(k) = \Oh\left(k^{\frac 1 2 - \eps}\right)$,
it is \emph{W[1]}-hard to distinguish whether for a~given instance of~\mmax $(k,p)$-\dmc:
\begin{enumerate}
\itemsep0em
    \item there is a cut of size $p$ that separates all $k$ terminal pairs, or
    \item all cuts of size $p$ separate at most $\frac{1}{\alpha(k)}\cdot k$ terminal pairs.
\end{enumerate}
\end{restatable}

When restricted to polynomial running time, the lower bound of $\Omega(k^{\frac 1 2 - \eps})$ can be improved to $\Omega(|E(G)|^{\frac 1 2 - \eps})$, however unlike the case of $k$-\so, this time the reduction is~polynomial and we need to assume only 
 NP $\not\subseteq$ co-RP\footnote{A problem is in co-RP if it admits a polynomial-time false-biased  algorithm, i.e., an~algorithm which is always correct for YES-instances and for NO-instances returns the correct answer with probability greater than some constant.}.
 
\begin{restatable}{theorem}{polydmc}
Assuming \emph{NP} $\not\subseteq$ \emph{co-RP}, for any $\eps > 0$ and function $\alpha(m) = \Oh\left(m^{\frac 1 2 - \eps}\right)$, there is no polynomial-time algorithm that, given an instance $(G,\mathcal{T},p),\, |\mathcal{T}|=k,\, |E(G)|=m$, of \mmax \dmc, distinguishes between the following cases:
\begin{enumerate}
\itemsep0em
    \item there is a cut of size $p$ that separates all $k$ terminal pairs, or
    \item all cuts of size $p$ separate at most $\frac{1}{\alpha(m)}\cdot k$ terminal pairs.
\end{enumerate}
\end{restatable}

As far as we know, the approximation status of this variant has not been studied yet. 
If~we want to minimize the number of removed edges to separate all terminal pairs
or minimize the ratio of the cutset size to the number of separated terminal pairs (this problem is known as \textsc{Directed Sparsest Multicut}), those cases admit a polynomial-time $\tilde{\Oh}(n^{\frac{11}{23}})$-approximation algorithm~\cite{multicut-ub} and a lower bound of $2^{\Omega(\log^{1-\eps} n)}$~\cite{multicut-lb}.
Since $\frac{11}{23} < \frac{1}{2}$ and $n \le m$, the maximization variant with a~hard constraint on the cutset size turns out to be harder.

In the undirected case \textsc{$p$-Multicut} is FPT, even when parameterized only by the size of~the~cutset $p$ and allowing arbitrarily many terminals~\cite{multicut-fpt}.
This is in contrast with the directed case, which becomes W[1]-hard already for 4 terminals, when parameterized by $p$.
It is worth mentioning that $k$-\so and $p$-\dmc were proven to be W[1]-hard with a~similar gadgeting machinery~\cite{pilipczuk-multicut}.

\paragraph{Organization of the paper}
We~begin with the necessary definitions in Section~\ref{sec:prelim}.
As our gap amplification technique is arguably the most innovative ingredient of the paper, we precede the proofs with informal Section~\ref{sec:technique}, which introduces the ideas gradually.
It is followed by the detailed constructions for $k$-\so  in Section \ref{sec:so}
and for \mmax $(k,p)$-\dmc in~Section \ref{sec:dmc}.
Each contains a~self-reduction lemma and applications to polynomial and FPT running time.
Both self-reductions are based on a~probabilistic argument described in Section~\ref{sec:correct}.
Finally, in Section \ref{sec:w1} we prove W[1]-completeness of $k$-\so.
That part of the paper is self-contained and can be read independently of the previous sections.

\section{Preliminaries}
\label{sec:prelim}
\paragraph*{Fixed parameter tractability}

A parameterized problem instance is
    created by associating an integer parameter \(k\) with an input instance.
    Formally, a parameterized language is a subset of $\Sigma^* \times \mathbb{N}$.
    We
    say that a language (or a problem) is \emph{fixed parameter tractable} (FPT{}) if it admits an algorithm solving an
    instance \((I, k)\) (i.e., deciding if it belongs to the language) in running time
    \(f(k)\cdot |I|^{\Oh(1)}\), where \(f\) is a computable function.
    Such a~procedure is called an \emph{FPT algorithm} and we say concisely that it runs in \emph{FPT time}.
    A language belongs to the broader class XP if it admits an algorithm with running time of the form
    \( |I|^{f(k)}\).
    
    There is no widely recognized class describing problems which admit randomized FPT algorithms.
    Instead of defining such a~class, we will directly use
    the notion of a~\emph{false-biased} algorithm, which is always correct for YES-instances and for NO-instances returns the correct answer with probability greater than some constant (equivalently, when the~algorithm returns \emph{false} then it is always correct).
    Similarly, a~\emph{true-biased} algorithm is always correct for NO-instances but may be wrong for YES-instances with bounded probability.
    A false-biased (resp. true-biased) FPT algorithm satisfies the condition above and runs in FPT time.

    To argue that a
    problem is unlikely to be FPT{}, we use parameterized reductions analogous
    to those employed in the classical complexity theory. Here, the concept of
    {W}-hardness replaces NP-hardness, and we need not only to
    construct an equivalent instance in FPT{} time, but also ensure that the parameter in the new instance depends only on the
    parameter in~the original instance.
If there exists a parameterized reduction from a W[1]-hard problem (e.g., $k$-\textsc{Clique}) to another problem \(\Pi \), then the problem \(\Pi \) is
    {W[1]}-hard as well. This provides an argument that \(\Pi \) does not admit an algorithm with
    running time \(f(k)\cdot |I|^{\Oh(1)}\)
    under the assumption that FPT $\ne$ W[1].

\paragraph*{Approximation algorithms and gap problems}
We define an optimization problem (resp. parameterized optimization problem) as a task of optimizing function $L \rightarrow \mathbb{N}$, where $L \subseteq \Sigma^*$ (resp. $L \subseteq \Sigma^* \times \mathbb{N}$), representing the value of the optimal solution.
An $\alpha$-approximation algorithm for a~maximization task must return a solution of value no less than the optimum divided by $\alpha$ (we follow the convention that $\alpha > 1$).
The approximation factor $\alpha$ can be a~constant or it can depend on the input size.
In the most common setting, the running time is required to be polynomial.
An~FPT approximation algorithm works with a parameterized optimization problem and is required to run in FPT time.
It is common that its approximation coefficient can depend on the parameter.

A gap problem (resp. parameterized gap problem) is given by two disjoint languages $L_1, L_2 \subseteq \Sigma^*$ (resp. $L_1, L_2 \subseteq \Sigma^* \times \mathbb{N}$).
An algorithm should decide whether the input belongs to $L_1$ or to $L_2$.
If neither holds, then the algorithm is allowed to return anything.
Usually $L_1, L_2$ are defined respectively as the sets of instances (of an optimization problem) with a~solution of value at least $C_1$ and instances with no solution with value greater than $C_2$.
An $\alpha$-approximation algorithm with $\alpha < \frac{C_1}{C_2}$ can distinguish $L_1$ from $L_2$, therefore hardness of an approximation task is implied by hardness for the~related gap problem.

\paragraph{Problem definitions}
We now formally describe the problems we work with.
Since we consider parameterized algorithms it is important to specify how we define the parameter of an instance.

A mixed graph is a triple $(V, A, E)$, where $V$ is the vertex set, $A$ is the set of directed edges, and $E$ stands for the set of undirected edges.
An~orientation of a~mixed graph is given by replacing each undirected edge $uv \in E$ with one of the directed ones: $(u,v)$ or $(v,u)$.
This creates a~directed graph $\tilde{G} = (V, A \cup \tilde{E})$, where $\tilde{E}$ is the set of newly created directed edges.
We assume that $uv \notin E$ for each $(u,v) \in A$, so $\og$ is always a simple graph.

\problembox{$k$-\so}{
mixed graph $G = (V, A, E)$, list of terminal pairs $\mathcal{T} = ((s_1,t_1),\dots,(s_k,t_k))$, }{
$k$}{
find an orientation $\og$ of $G$ that maximizes the number of pairs $(s_i,t_i)$, such that $t_i$ is reachable from $s_i$ in $\tilde{G}$
}

We add ,,\textsc{Max}'' to the name of the following problem in order to distinguish it from the more common version of \dmc, where one minimizes the number of edges in the cut.

\problembox{\textsc{Max $(k,p)$-Directed Multicut}}{
directed graph $G = (V, A)$, list of terminal pairs $\mathcal{T} = ((s_1,t_1),\dots,(s_k,t_k))$, 
integer $p$}{
$k+p$}{
find a subset of edges $A' \subseteq A$, $|A'| \le p$, in order to maximize the number of pairs $(s_i,t_i)$, such that
$t_i$ is unreachable from $s_i$ in $G \setminus A'$}


If a solution to either problem satisfies the condition for a particular terminal pair, we say that this pair is satisfied by this solution.
The decision versions of both problems ask whether there is a solution of value $k$, that is, satisfying all the terminal pairs.
We call such an instance fully satisfiable, or a YES-instance, and a NO-instance otherwise.
For the sake of proving approximation hardness we introduce the gap versions: $q$-\textsc{Gap} $k$-\so and $q$-\textsc{Gap} \textsc{Max $(k,p)$-Directed Multicut}, where we are promised that the value of the optimal solution is either $k$ or at most $\frac{k}{q}$, and we have to distinguish between these cases.

When referring to non-parameterized problems, we drop the parameters in the problem name.
We use notation $[n] = \{1,2,\dots, n\}$. All logarithms are 2-based.

\section{The gap amplification technique}
\label{sec:technique}

We begin with an informal thought experiment that helps to understand the main ideas behind the reduction.
For an instance $(G, \mathcal{T} = ((s_1,t_1),\dots,(s_k,t_k)))$ of $k$-\so
we refer to~the~vertices $s_i, t_i \in V(G)$ as $G\{s,i\}$ and $G\{t,i\}$.
We want to construct a larger instance $(H, \mathcal{T}_H)$ so that if $(G, \mathcal{T})$ is fully satisfiable then $(H, \mathcal{T}_H)$ is as well, but otherwise the maximal fraction of satisfiable pairs in $(H, \mathcal{T}_H)$ is strictly less than $\frac{k-1}{k}$.
Consider $k$ vertex-disjoint copies of the original instance: $(G_1, \mathcal{T}_1), (G_2, \mathcal{T}_2),\dots, (G_k, \mathcal{T}_k)$, that will be treated as the first layer.
Assume that $(G, \mathcal{T})$ is a NO-instance (i.e., one cannot satisfy all pairs at once), so for any orientation of the copies $\og_1, \og_2,\dots, \og_k$, there is a tuple $(j_1, j_2,\dots, j_k)$ such that $\og_i\{t, j_i\}$ is unreachable from $\og_i\{s, {j_i}\}$ in $\og_i$.
Suppose for now that we have fixed the values of $(j_i)$, even before we have finished building our instance.

Let $R = (r_1, r_2,\dots, r_k)$ be a tuple sampled randomly from $[k]^k$.
We connect the sinks in the first layer to the sources in another copy of the same instance -- let us refer to it as $(G_R, \mathcal{T}_R)$.
We add a directed edge from $G_i\{t, r_i\}$ to $G_R\{s, i\}$ for each $i \in [k]$, thus connecting a random sink of $G_i$ to the source $G_R\{s,i\}$,
as shown in Figure~\ref{fig:so1}.
We refer to the union of all $k+1$ copies of $G$ with $k$ added connecting edges as the graph $H$.
We~define $\mathcal{T}_H = ((G_1\{s, r_1\}, G_R\{t, 1\}),\dots,(G_k\{s, r_k\}, G_R\{t, k\}))$, so we want to satisfy those $k$ terminal pairs that got connected randomly.
Let $X$ be a~random variable equal to the value of the optimal solution for $(H, \mathcal{T}_H)$
under the restriction that the solution orients $G_1, G_2 \dots, G_k$ as
$\og_1, \og_2,\dots, \og_k$.
What would be the expected value of $X$?

\begin{figure}
\centering
\includegraphics[scale=.9]{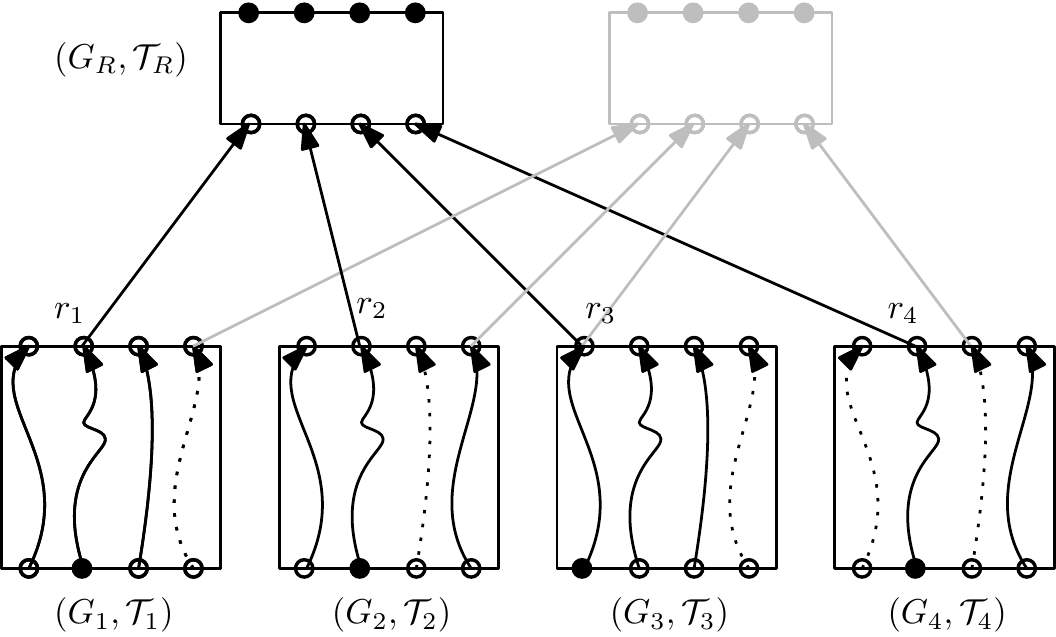}
\caption{ The construction of $(H, \mathcal{T}_H)$ for $k=4$.
The black copy $(G_R, \mathcal{T}_R)$ is connected according to the choice of $R = (r_1, r_2,r_3, r_4) = (2, 2, 1, 2)$ and the terminals (sources and sinks) are marked with black dots.
The dotted lines indicate which pairs of terminals are unreachable in each $\og_i$ and we can set $(j_1, j_2,j_3, j_4) = (4, 3, 4, 3)$ (another valid choice is $(4,3,4,1)$).
We have $r_i \ne j_i$ for each $i$, so $Y = 4$ but $X \le 3$.
The grey copy illustrates another random choice of connection: $R' = (r_1, r_2,r_3, r_4) = (4, 4, 1, 3)$ and $X \le Y = 2$.
If we wanted to construct $(H, \mathcal{T}_H)$ emulating the whole probabilistic space, we would include all $4^4$ copies of $(G_R, \mathcal{T}_R)$ for all choices of $R = (r_1, r_2,r_3, r_4)$ in the same way as the black and the grey copy.
}
\label{fig:so1}
\end{figure}

Let $Y$ denote another random variable being the number of
indices $i$ for which $r_i \ne j_i$. 
By~linearity of expectation we have $\ex Y \le k-1$.
It holds that $X \le Y$ and so far we still have only the bound $\ex X \le k-1$.
However, with probability $(\frac{k-1}{k})^{k}$ we have $r_i \ne j_i$ for all $i$, therefore $Y = k$, but we cannot connect all pairs within $G_R$ (because it is a copy of a NO-instance), so $X \le k-1$.
This means that $\ex(Y-X) \ge (\frac{k-1}{k})^{k}$ and so $\ex X \le k - 1 - (\frac{k-1}{k})^{k}$: the gap has been slightly amplified.

Of course in the proper reduction we cannot fix the orientation before adding the connecting edges.
However, we can afford an exponential blow-up with respect to $k$.
We can include in the second layer the whole probabilistic space, that is, $k^k$ copies of $(G, \mathcal{T})$ (rather than a single $(G_R, \mathcal{T}_R)$), each connected to the first layer with respect to a~different tuple $(r_1, r_2,\dots, r_k)$, thus creating a~large instance  $(H, \mathcal{T}_H)$ with $k^k \cdot k$ terminal pairs (see~Figure~\ref{fig:so1}).
For any orientation of ${H}$ the fraction of satisfied terminal pairs equals the average over the fractions for all $k^k$ groups of terminal pairs,
so we can emulate the construction above without fixing $(j_1, j_2,\dots, j_k)$.
The maximal fraction of satisfiable terminal pairs in such (no longer random) $(H, \mathcal{T}_H)$ would be the same as~before, that is, $\frac{\ex X}{k} < k-1$.
However, the smaller instance we create the better lower bounds we get, so we will try to be more economical while constructing  $(H, \mathcal{T}_H)$.

An important observation is that we do not have to include all $k^k$ choices of $R$ in the construction.
We just need a sufficient combination of them, so that the gap amplification occurs for any choice 
of $(j_1, j_2,\dots, j_k)$. 
This can be ensured by picking just $k^{\Oh(1)}$ choices of $R$ and using an argument based on the Chernoff bound (see Section~\ref{sec:so} for a detailed construction).

We can iterate this construction by treating $(H, \mathcal{T}_H)$ as a new input and amplifying the gap further in each step.
In further steps we need to add an exponential number of copies to the new layer, even when compressing the probabilistic space as above.
This is why we get an exponential blow-up with respect to $k$ and we need to work with a~parameterized hardness assumption, even for ruling out polynomial-time approximations.

The construction for \mmax $(k,p)$-\dmc is simpler because the layer stacking does not have to be iterated.
Therefore to achieve polynomial-time hardness it suffices to assume that NP $\not\subseteq$ co-RP.
The phenomenon that both problems admit such strong self-reduction properties can be explained by the fact that when dealing with directed reachability one can compose instances sequentially, which is the first step in both reductions.

\section{Inapproximability of \so}
\label{sec:so}


We first formulate the properties of the construction and discuss how they are used to prove the claims of the paper.
Then, we focus on the construction itself in Section~\ref{sec:correct}.
Let $S(G,\mathcal{T})$ denote the maximal number of pairs that can be satisfied by~some~orientation in the instance $(G,\mathcal{T})$.

\begin{lemma}\label{thm:so-gap}
There is a procedure that, for an instance $(G,\mathcal{T}_G)$ of \so and parameter $q$, constructs a new instance $(H,\mathcal{T}_H),\, k_0=|\mathcal{T}_H|$, such that:
\begin{enumerate}
\itemsep0em
    \item $k_0 = 2^{q^{\Oh(k)}}$,
    \item $|V(H)| \le |V(G)| \cdot k_0^2$, 
    \item if $S(G,\mathcal{T}_G) = k$, then $S(H,\mathcal{T}_H) = k_0$ always (Completeness),
    \item if $S(G,\mathcal{T}_G) < k$, then $S(H,\mathcal{T}_H) \le \frac{1}{q}\cdot k_0$ with probability at least $\frac{1}{k_0}$ (Soundness).
\end{enumerate}
The randomized construction runs in time proportional to $|H|$.
It can be derandomized within running time $f(k,q)\cdot |G|$.
\end{lemma}


These properties are proven at the end of Section~\ref{sec:correct}.
It easily follows from them that the gap can get amplified to any constant~$q$.
It is more complicated to rule out a~superconstant approximation factor, e.g. $\alpha(k) = \log\log k$, because we would like to apply the hypothetical algorithm to the instance $(H,\mathcal{T}_H)$ with parameter $k_0$ depending on $k$ and $q$, so
we additionally need to adjust $q$ so that $\alpha(k_0(k,q)) \le q$.

\mainso*
\begin{proof}
We are going to reduce the exact version of $k$-\so, which is W[1]-hard, to the version with a sufficiently large gap, with Lemma~\ref{thm:so-gap}.
For a~fixed $k$
we can bound $k_0$ by a~function of $q$: $k_0(q) \le 2^{q^{c\cdot k}}$ for some constant $c$.
On the other hand, $k_0(q) \rightarrow \infty$.
Given an instance $(G,\mathcal{T}_G)$ we use Lemma~\ref{thm:so-gap} with~$q$ large enough, so that $\beta(k_0(q))\cdot c\cdot k \le 1$ (this is possible because $\beta(k) \rightarrow 0$).
Such a dependency $q(k)$ is also a computable function.
We get $\alpha(k_0) = (\log k_0)^{\beta(k_0)} \le q^{c\cdot k \cdot \beta(k_0)} \le q$.

We have obtained a new instance $(H,\mathcal{T}_H)$
of $k_0$-\so  of size $f(k)\cdot |V(G)|$ and $k_0$ being a function of $k$.
If the original instance is fully satisfiable then the same holds for $(H,\mathcal{T}_H)$ and otherwise
$S(H,\mathcal{T}_H) \le \frac{1}{q}\cdot k_0 \le \frac{1}{\alpha(k_0)}\cdot k_0$, which finishes the reduction.
\end{proof}

If we restrict the running time to be purely polynomial,
we can slightly strengthen the lower bound, i.e., replace $k$ with $n$ in the approximation factor, while working with
a~similar hardness assumption.
To make this connection, we observe that in order to show that a problem is in FPT, it suffices to solve it in polynomial time for some superconstant bound on the parameter.

\begin{fact}\label{lem:xp}
Consider a parameterized problem $\Pi \in$ \emph{XP} that admits a polynomial-time algorithm (resp. false-biased polynomial-time algorithm) for the case $f(k) \le |I|$, where $f$ is some computable function.
Then $\Pi$ admits an \emph{FPT} algorithm (resp. false-biased \emph{FPT} algorithm).
\end{fact}
\begin{proof}
Since $\Pi \in$ XP, it admits a~deterministic algorithm with running time $|I|^{g(k)}$.
Whenever $f(k) \le |I|$, we~execute the polynomial-time algorithm.
Otherwise we can solve it in time $f(k)^{g(k)}$.
\end{proof}

\polyso*
\begin{proof}
Suppose there is such an algorithm with approximation factor $\alpha(n) = (\log n)^{\beta(n)}$. 
Let $\beta^*(\ell)$ be the smallest integer $L$, for which $\beta(L) \le \frac 1 \ell$.
The function $\beta^*$ is well defined and computable, because $\beta$ is computable.
Again, for fixed $k$ and some constant $c$ we have dependency $k_0(q) \le 2^{q^{c\cdot k}}$.

We are going to use the polynomial-time algorithm to solve $k$-\so in randomized $f(k)\cdot n^{\Oh(1)}$ time, which would imply the claim.
Since the problem is in XP~\cite{xp},
by Fact~\ref{lem:xp} it suffices to solve instances satisfying $\beta^*((c k)^2) \le n$ in polynomial time.
We~can thus assume $(c k)^2\cdot \beta(n) \le 1$, or equivalently
$ck \cdot \beta(n) \le \sqrt{\beta(n)}$.

Given an instance of $k$-\so, we apply Lemma~\ref{thm:so-gap} with $q = (2\log n)^{\beta(n)}$. 

$$ k_0 \le 2^{q^{ck}} = 2^{(2\log n)^{ck\cdot \beta(n)}} \le 2^{(2\log n)^{\sqrt{\beta(n)}}} = n^{o(1)},$$
$$\alpha(|V(H)|) \le \alpha(n\cdot k_0^2) =  \alpha(n^{1 + o(1)}) = ((1+o(1))\cdot \log n)^{ \beta(n)} .$$

For large $n$ we obtain $\alpha(|V(H)|) \le q$.
Since $|V(H)| \le n\cdot k_0^2 = n^{1+o(1)}$, the size of the new instance of polynomially bounded and thus the randomized construction takes polynomial time.
If~we started with a~YES instance, we always produce a~YES instance, and otherwise
$S(H,\mathcal{T}_H) \le \frac{1}{q}\cdot k_0 \le \frac{1}{\alpha(|V(H)|)}\cdot k_0$ with probability at least $\frac{1}{k_0}$, so we need to repeat the procedure $k_0 = n^{o(1)}$ times to get a~constant probability of~being correct.
A~hypothetical algorithm distinguishing these cases would therefore entail a~false-biased polynomial-time algorithm for \so
for the case $\beta^*((c k)^2) \le n$.
The claim follows from Fact~\ref{lem:xp}.
\end{proof}

\subsection{The gap amplifying step}
\label{sec:correct}



We are going to present a formal construction of the argument sketched in Section~\ref{sec:technique}.
Given fixed orientations of the $k$ copies of $G$, we are able to randomly sample $k$ sinks and insert additional edges so that the expected optimum of the new instance is sufficiently upper bounded.
We want to reverse this idea, so we could randomly sample a moderate number of additional connections once to ensure the upper bound works for any orientation.
To this end, we need some kind of a~hashing technique to mimic the behaviour of the probabilistic space with a~structure of moderate size.
Examples of such constructions are (generalized) universal hash families
\cite{hash, hash2, vazirani-phd}
or expander random walk sampling~\cite{expander}.
Even though the construction presented below is relatively simple, we are not aware of any occurrences of it in the literature.

For a~set $X_1$ and a multiset $X_2$, we write $X_2 \subseteq X_1$ if every element from~$X_2$ appears in~$X_1$.
Let $U(X)$ denote the uniform distribution over a finite multiset $X$.
In~particular, each distinct copy of the same element in $X$ has the same probability of being chosen: $\frac{1}{|X|}$.

\begin{definition}
For a family $\mathcal{F}$ of functions $X \rightarrow [0,1]$,
a $\delta$-biased sampler family is a multiset $X_H \subseteq X$, such that
for every $f \in \mathcal{F}$ it holds

$$ \left| \ex_{x \sim U(X_H)} f(x) - \ex_{x \sim U(X)} f(x) \right| \le \delta. $$
\end{definition}

\begin{lemma}
\label{lem:hash}
For a given $X,\, \mathcal{F}$, and $\delta > 0$,
a sample of $\Oh( \delta^{-2}\log(| \mathcal{F}|))$ elements from $X$ (sampled independently with repetitions)
forms a~$\delta$-biased sampler family with probability at least~$\frac{1}{2}$.
\end{lemma}
\begin{proof}
We sample independently $M = 10\cdot \delta^{-2}\log(| \mathcal{F}|)$ elements from $X$ with repetitions.
For the sake of analysis, note that this is a single sample from the space $\Omega_{X,M}$ being the family of all $M$-tuples of elements from $X$, equipped with a uniform distribution.
Let $X_H$ denote the random multiset of all elements in this $M$-tuple.
For each $f \in \mathcal{F}$ we define $A_f \subseteq \Omega_{X,M}$ as the family of tuples for which $\left| \ex_{x \sim U(X_H)} f(x) - \ex_{x \sim U(X)} f(x) \right| > \delta$.
For a fixed $f$ we apply the Hoeffding's inequality.

$$  \mathbb{P}(A_f) = \mathbb{P}\left(\bigg|\ex_{x \sim U(X_H)} f(x) - \ex_{x \sim U(X)} f(x) \bigg| > \delta\right) \le 2\exp(-2\delta^2 M).$$
For our choice of $M$ this bound gets less than $\frac{1}{2 | \mathcal{F}|}$.
By union bound, the probability that $X_F$ is not a~$\delta$-biased sampler family is
$\mathbb{P}\left(\bigcup_{f \in \mathcal{F}} A_f\right) \le \sum_{f \in \mathcal{F}} \mathbb{P}(A_f) \le \frac{1}{2}$.
The claim follows.
\end{proof}

The rest of this section is devoted to proving Lemma~\ref{thm:so-gap}.
We keep the concise notation from Section~\ref{sec:technique}: for~an~instance $(G, \mathcal{T})$,  $\mathcal{T} = ((s_1,t_1),\dots,(s_k,t_k))$ of $k$-\so
we refer to~the~vertices $s_i, t_i \in V(G)$ as $G\{s,i\}$ and $G\{t,i\}$

\paragraph{Building the layers}
Given an instance $(G, \mathcal{T}_G)$, our aim is to build a larger instance, so that if $S(G,\mathcal{T})=k$ then the new one is also fully satisfiable, but otherwise the maximal fraction of terminal pairs being simultaneously satisfiable in the new instance is at most $\frac{1}{q}$.

We inductively construct a family of instances $(H^i, \mathcal{T}^i)_{i=1}^M$ with $(H^1, \mathcal{T}^1) = (G, \mathcal{T}_G)$.
Let $k_i = |\mathcal{T}^i|$ and $p_i$ indicate the number of copies of $(G, \mathcal{T}_G)$ in the last layer (to be explained below) of $(H^i, \mathcal{T}^i)$.
We will have that $p_1 = 1$ and $k_i = |\mathcal{T}^i| = k\cdot p_i$.
We construct $(H^{i+1}, \mathcal{T}^{i+1})$ by taking $k$ vertex-disjoint copies of the $i$-th instance, denoted $(H^i_1, \mathcal{T}^i_1),\dots, (H^i_k, \mathcal{T}^i_k)$ and forming a new layer of copies of $(G, \mathcal{T}_G)$ which will be randomly connected to the $i$-th layer through directed edges.
Therefore graph $H^{i+1}$ will have $i+1$ layers of copies of $G$.

Let $\mathcal{R} = [k_i]^k$ be the family of $k$-tuples $(r_1, r_2, \dots r_k)$ with elements from the set $[k_i]$.
We sample a random tuple $R = (r_1, r_2, \dots r_k)$ from $\mathcal{R}$ and create a new copy of the original instance -- let us refer to it as $(G_R, \mathcal{T}_R)$.
We add a directed edge from $H^i_j\{t, r_j\}$ to $G_R\{s, j\}$ for each $j \in [k]$, thus connecting a random sink of $H^i_j$ to the source $G_R\{s,i\}$.
We insert $k$ pairs to $\mathcal{T}^{i+1}$: $(H^i_1\{s, r_1\}, G_R\{t, 1\}), \dots, (H^i_k\{s, r_k\}, G_R\{t, k\})$.
We iterate this subroutine $p_{i+1} = \Oh(k^4q^{2k}p_i)$ times (a~derivation of this quantity is postponed to Lemma~\ref{lem:layer}),
as shown in Figure~\ref{fig:so}.

\begin{figure}
\centering
\includegraphics[scale=0.75]{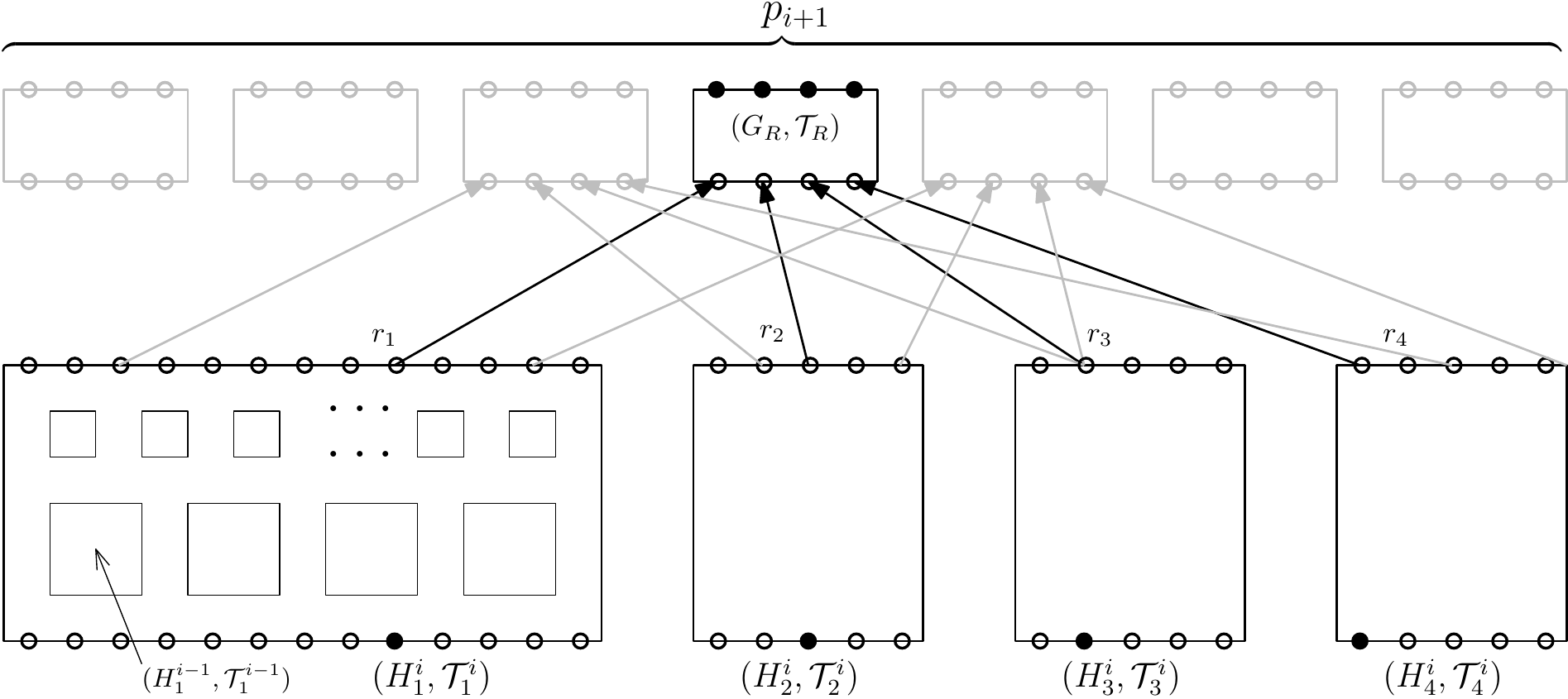}
\caption{ The construction of $(H^{i+1}, \mathcal{T}^{i+1})$ for $k=4$.
The first copy of $(H^{i}, \mathcal{T}^{i})$ is shown in greater detail to highlight the recursive construction, with circles representing its terminal pairs.
The new layer consists of $p_{i+1}$ copies of $(G, \mathcal{T})$.
The vector $R = (r_1, r_2,r_3, r_4) \in [p_i]^4$ indicates which sinks are connected to the sources in $(G_R, \mathcal{T}_R)$.
The black dots show $4$ terminal pairs associated with this copy.
For the sake of legibility only the edges incident to $G_R$ (the black ones) and neighboring copies (the grey ones) are sketched.
}\label{fig:so}
\end{figure}

\begin{lemma}
\label{lem:layer}
Let $y_i = S(H^i,\mathcal{T}^i)\, /\, k_i$ be the maximal fraction of terminal pairs that can be simultaneously satisfied in $(H^i,\mathcal{T}^i)$.
Suppose that $S(G,\mathcal{T}_G) < k$ and $y_i \ge \frac 1 q$.
Then with probability at least $\frac 1  2$ it holds $y_{i+1} \le y_i - \frac{1}{2k}\cdot q^{-k}$.
\end{lemma}
\begin{proof}
First observe that for each $(s_j, t_j) \in \mathcal{T}^i$, any $(s_j, t_j)$-path in $(H^i,\mathcal{T}^i)$ runs through $i$ unique copies of $(G,\mathcal{T}_G)$.
Therefore $(s_j, t_j)$-pair is satisfied only if the corresponding $i$~terminal pairs in those copies are satisfied.
Recall that we have connected each copy of~$(G,\mathcal{T}_G)$ in the $i$-th layer to the terminals from the previous layer according to a~random tuple $R = (r_1, r_2, \dots r_k) \in \mathcal{R}$.

We will now analyze how the possible orientations $\tilde{H}^i_1, \dots, \tilde{H}^i_k$ influence the status of the terminal pairs in $\mathcal{T}^{i+1}$.
Let $C_j \subseteq [k_i]$ encode which of the terminal pairs are reachable in $\tilde{H}^i_j$, that is, $C_j = \{\ell \in [k_i] \colon \tilde{H}^i_j\{t,\ell\}  \text{ is reachable from }  \tilde{H}^i_j\{s,\ell\}\}$.
A tuple $C = (C_1, \dots C_k)$ is called a~\emph{configuration} and we denote the family of all feasible configurations as  $\mathcal{C}$.
We have $|\mathcal{C}| \le  (2^{k_i})^{k} = 2^{p_ik^2}$.
For a configuration $C \in \mathcal{C}$ let $f_C: \mathcal{R} \rightarrow [0,1]$ be a function describing the maximal fraction of satisfiable terminal pairs from those with sinks in $G_R$ connected through a~tuple $R \in \mathcal{R}$.
Note that $f_C(R)$ depends on $C, R$ and the best possible orientation of $G_R$, whereas it is oblivious to the rest of the structure of $(\tilde{H}^{i}_j)_{j=1}^k$.
We can thus think of~$C$ as an interface between the first $i$ layers and $G_R$.

For fixed orientations $\tilde{H}^i_1, \dots, \tilde{H}^i_k$, and therefore fixed configuration $C \in \mathcal{C}$, we estimate the expected value of $f_C$.
Let $Y_j$ be a random Boolean variable indicating that $\tilde{H}^i_j[t,r_j]$  is reachable from $\tilde{H}^i_j[s,r_j]$ for a random $(r_1, r_2, \dots r_k) \in \mathcal{R}$, i.e., that $r_j \in C_j$.
Let $Y = \sum_{j=1}^k Y_j / k$, so that we always have $f_C(R) \le Y$.

Let $c_j = |C_j| / k_i$.
By linearity of expectation $\ex Y = \sum_{j=1}^k c_j / k$ and by~the~assumption $c_j \le  y_i$ for all $j \in [k]$.
However, with probability $\prod_{j=1}^k c_j$ we have $r_j \in C_j$ for all $j$, so $Y = 1$, but we cannot connect all pairs within $(G_R, \mathcal{T}_R)$ (since we assumed $S(G,\mathcal{T}_G) < k$), so $f_C(R) \le 1-\frac 1 k$.
This means that
$$
\ex\left(f_C(R)-Y\right) \ge \frac{1}{k} \cdot \prod_{j=1}^k c_j, \quad \text{ and so } \quad
\ex f_C(R) \le  \frac{1}{k} \cdot \left(\sum_{j=1}^k c_j - \prod_{j=1}^k c_j \right).
$$


The quantity $\sum_{j=1}^k c_j - \prod_{j=1}^k c_j$ can only increase when we increase some $c_\ell$, because the $\ell$-th partial derivative is $1 - \prod_{j \in [k], \, j \neq \ell} c_j \ge 0$.
By the assumption $c_j \le y_i$ and $y_i \ge \frac{1}{q}$, hence
$$
\ex f_C(R) \le  \frac{1}{k} \cdot \left(\sum_{j=1}^k y_i - \prod_{j=1}^k y_i \right) \le y_i - \frac{1}{k} \cdot q^{-k}.
$$

Now we apply Lemma~\ref{lem:hash} for $\mathcal{F} = \{f_C : C \in \mathcal{C}\}$ and $\delta =   \frac{1}{2k}\cdot q^{-k}$ to argue that for our choice of $p_{i+1}$ -- the number of copies in the last layer -- the estimation works for all $C \in \mathcal{C}$ at~once.
The quantity $M = \Oh(\delta^{-2}\log(| \mathcal{F}|))$ in Lemma~\ref{lem:hash} becomes
$\Oh(k^4q^{2k}p_i)$, which is exactly as we defined $p_{i+1}$.
We have sampled $p_{i+1}$ tuples from $\mathcal{R}$ (let us denote this multiset as $\mathcal{R}_H \subseteq \mathcal{R}$) and added a~copy $(G_R, \mathcal{T}_R)$ for each $R \in \mathcal{R}_H$.
For a~fixed $C \in \mathcal{C}$, the maximal fraction of satisfiable terminal pairs in $(H^{i+1},\mathcal{T}^{i+1})$ equals the average of $f_C(R)$ over $R \in \mathcal{R}_H$.
By~Lemma~\ref{lem:hash}
we know that, regardless of the choice of $C$, this quantity is at most 

$$\ex_{R\sim U(\mathcal{R}_H)} f_C(R) \le \ex_{R\sim U(\mathcal{R})} f_C(R) +  \frac{1}{2k}\cdot q^{-k} \le y_i -  \frac{1}{2k}\cdot q^{-k},$$
with probability at least $\frac 1 2$ (that is, if we have chosen $\mathcal{R}_H$ correctly).
Since the upper bound works for all $C \in \mathcal{C}$ simultaneously, the claim follows. 

\end{proof}

\begin{proof}[Proof of Lemma~\ref{thm:so-gap}]
We define $(H,\mathcal{T}_H) = (H^B,\mathcal{T}^B)$ for $B = 2kq^k$.
The completeness is straightforward:
if $S(G, \mathcal{T}_G) = k$, then
we can orient all copies of $G$ so that $G\{t,j\}$ is always reachable from $G\{s,j\}$ and
each requested path in $S(H^i,\mathcal{T}^i)$ is given as a concatenation of~respective paths in $B$ copies of $G$.

To see the soundness, suppose that $S(G, \mathcal{T}_G) < k$.
The sequence $(y_i)_{i=0}^B$ is non-increasing and
the value of $y_i$ is being decreased by at least $\frac{i}{2k}\cdot q^{-k}$ in each iteration, as long as $y_i \ge \frac 1 q$, due to Lemma~\ref{lem:layer}.
Therefore after $B = 2kq^k$ iterations we are sure to have $y_B \le \frac 1 q$.


To estimate the size of the instance, recall that we have  $k_i =p_ik$ and $p_{i+1} = \Oh(k^4q^{2k}p_i)$.
We can assume $q \ge 2$ and so $k \le q^k$.
For $B = 2kq^k$, the value of~$p_B$ becomes $\left(k^4q^{2k}\right)^{\Oh(kq^k)} = 2^{q^{\Oh(k)}\Oh(\log q)} = 2^{q^{\Oh(k)}}$.
The size of $V(H)$ is at most $k_B \cdot |V(G)|$ times the number of layers, which is $B$.
We trivially bound $B \le 2^B \le k_B$ to obtain $|V(H)| \le |V(G)| \cdot k_B^2$.

The presented construction is randomized because we randomly choose a~biased sampler family in each of the $B$ steps.
If we start with a~YES-instance, then we produce a~YES instance regardless of these choices, and otherwise we produce a~NO instance
with probability at~least $2^{-B} \ge \frac {1}{k_B}$.
The construction can be derandomized within running time $f(k,q)\cdot |V(G)|$ as follows.
In each application
of Lemma~\ref{lem:hash} the sizes of $X$ and $\mathcal{F}$ are $(k_i)^k$ and $2^{p_ik^2}$, respectively, and $\delta =   \frac{1}{2k}\cdot q^{-k}$, which
are all bounded by a function of $k$ and $q$.
Therefore instead of sampling a~biased sampler family, we can enumerate all $\Oh( \delta^{-2}\log(| \mathcal{F}|))$-tuples of~elements from~$X$
 and find one giving a~biased sampler family.
\end{proof}

\section{Inapproximability of \dmc}
\label{sec:dmc}

We switch our attention to the \textsc{Max $(k,p)$-Directed Multicut} problem, for which we provide a~slightly simpler reduction.
We keep the same convention as before: within graph $G$ we refer to sources and sinks $(s_i, t_i) \in \mathcal{T}$ shortly as $G\{s,i\}$, $G\{t,i\}$, and denote the maximal number of  terminal pairs separable in $(G,\mathcal{T})$ by~deleting $p$ edges by $S(G,\mathcal{T}, p)$.

\begin{lemma}\label{thm:dmc-gap}
There is a procedure that, for an instance $(G,\mathcal{T}_G,p),\, |\mathcal{T}_G| = 4$ of \dmc and parameter $q$, constructs a new instance $(H,\mathcal{T}_H,p_0),\, k_0=|\mathcal{T}_H|$, such that:
\begin{enumerate}
\itemsep0em
    \item $k_0 = \Theta(p\cdot q^{2}\log q)$,
    \item $p_0 = \Theta(p^2 \log q)$,
    \item $|E(H)| = |E(G)| \cdot p_0 +  \Oh(k_0\cdot p_0)$,
    \item if $S(G,\mathcal{T}_G,p) = 4$, then $S(H,\mathcal{T}_H,p_0) = k_0$ always (Completeness),
    \item if $S(G,\mathcal{T}_G,p) < 4$, then $S(H,\mathcal{T}_H,p_0) \le \frac 1 q\cdot k_0$ with probability at least $\frac 1 2$ (Soundness).
\end{enumerate}
The randomized construction takes time proportional to $|H|$. It can be derandomized in time \mbox{$f(p,q)\cdot |G|$}.
\end{lemma}
\begin{proof}
Consider $M = 3(p+1)\cdot \log q$ copies of $(G,\mathcal{T}_G)$,
denoted $(G_1,\mathcal{T}_1), \dots, (G_M,\mathcal{T}_M)$.
Let $\mathcal{R} =  [4]^M$ be the family of all M-tuples with values in $[4]$.
For a random sequence $R = (r_1, r_2, \dots r_{M}) \in  \mathcal{R}$,
we add a terminal pair $s_R, t_R$
and for each $i\in[M]$ we add directed edges $(s_R, G_i\{s,r_i\})$ and $(G_i\{t,r_i\}, t_R)$.
We repeat this subroutine $k_0 = \Theta(p\cdot q^{2}\log q)$ times and create that many terminals pairs as depicted in Figure~\ref{fig:dmc}.
We set the budget $p_0 = 3p(p+1)\cdot \log q$.

If $S(G,\mathcal{T}_G,p) = 4$,
then the budget suffices to separate all terminal pairs in all copies of~$(G,\mathcal{T}_G)$ (completeness).
Otherwise, one needs to remove at least $p+1$ edges from each copy of $(G,\mathcal{T}_G)$ to separate all 4 pairs so we can afford that in 
at most $3p\cdot\log q$ copies.
Therefore for any solution there
are at least $3\log q$ copies, where there is at least one terminal pair that is not separated.

Let $C_i \subseteq [4]$ represent information about the status of solution within $(G_i,\mathcal{T}_i)$: there is path from $G_i\{s, j\}$ to $G_i\{t, j\}$ only if $j \in C_i$.
A~tuple $C = (C_1, \dots,  C_M)$ is called a~configuration and we refer to the family of configurations induced by possible solutions as $\mathcal{C}$.
Clearly, $|\mathcal{C}| \le 16^M$.

Recall that each terminal pair can be represented by a tuple $R = (r_1, r_2, \dots r_{M}) \in  \mathcal{R}$ encoding through which terminal pair in $G_i$ a~path from $s_R$ to $t_R$ can go.
For a fixed configuration $C \in \mathcal{C}$, function $f_C: \mathcal{R} \rightarrow \{0,1\}$
is set to $1$ if the pair $s_R, t_R$ is separated, or equivalently: if for 
each $i \in [M]$ we have $r_i \not\in C_i$.
For $S(G,\mathcal{T}_G,p) < 4$ there are at least $3\log q$ copies of $G_i$ with $C_i \neq \emptyset$,
therefore $\ex_{R \sim U(\mathcal{R})} f_C(R) \le (\frac 3 4)^{3\log q} \le 2^{-\log (2q)} = \frac 1 {2q}$.

The size of $\mathcal{C}$ is at most $16^M = 2^{\Oh(p\log q)}$.
We apply Lemma~\ref{lem:hash} for $\mathcal{F} = \{f_C : C \in \mathcal{C}\}$ and $\delta =  \frac 1 {2q}$.
It follows that $\Oh( \delta^{-2}\log(| \mathcal{F}|)) = \Oh(p\cdot q^{2}\log q)$ random samples from $\mathcal{R}$ suffice to obtain a~rounding error of at most $ \frac 1 {2q}$ for all $C \in \mathcal{C}$ at once.
Therefore with probability at least $\frac 1 2$ we have constructed an~instance in which for any cutset of size $p_0$ (and thus for any configuration $C$) the fraction of separated terminal pairs is at most $\ex_{R \sim U(\mathcal{R})} f_C(R) + \frac 1 {2q} \le \frac 1 q$.
\end{proof}

\begin{figure}
\centering
\includegraphics[scale=0.75]{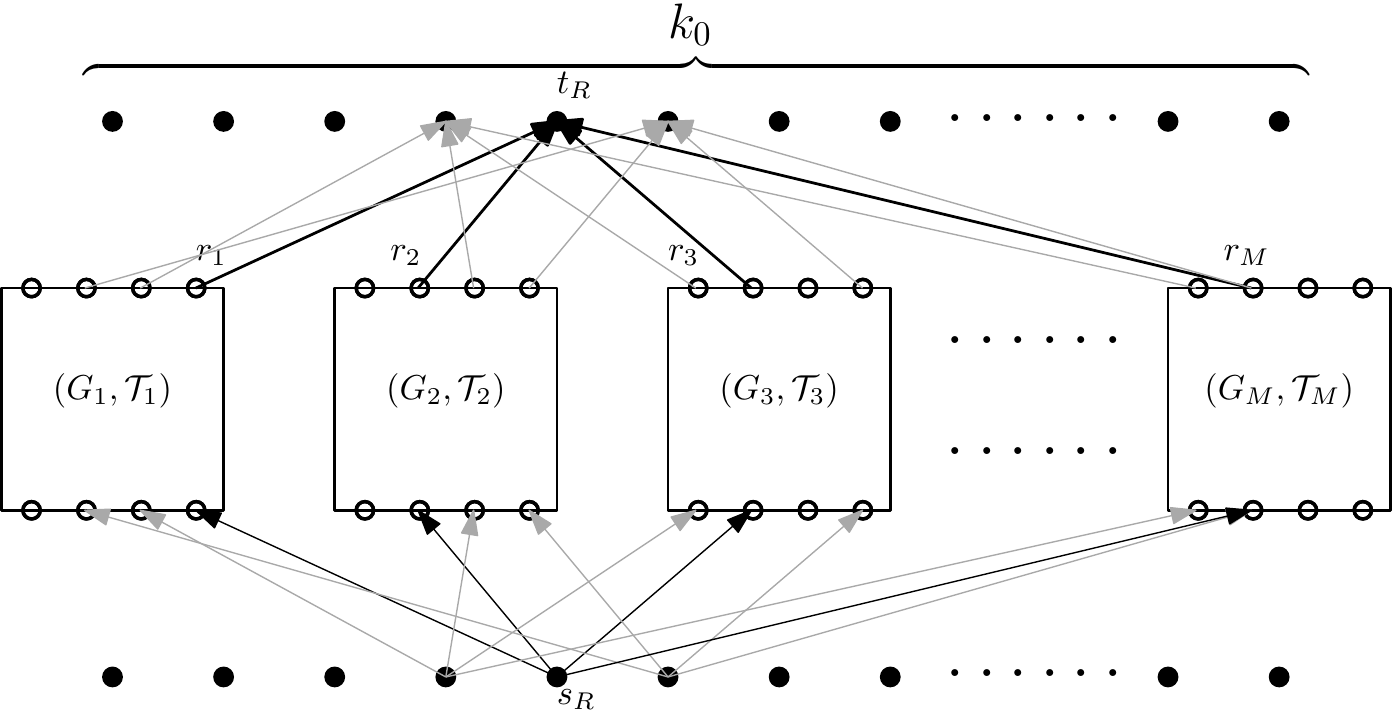}
\caption{ Building a new instance from $M$ parallel copies of $G$.
The tuple $R = (r_1, r_2, \dots r_{M})$ encodes through which nodes the $(s_R,t_R)$ pair is connected.
For the sake of legibility only the edges incident to $s_R, t_R$ (the black ones) and neighboring terminals (the grey ones) are sketched.
}\label{fig:dmc}
\end{figure}

\paragraph{Remark on derandomization} As before, if we allow exponential running time with respect to $p$ and $q$,
we can find a correct sampler family by enumeration and derandomize the reduction.
However, we cannot afford that in a polynomial-time reduction.
To circumvent this, observe that we upper bound the expected value of $f_C$ using independence of $3\log q$ variables.
We could alternatively take advantage of $\delta$-biased $\ell$-wise independent hashing~\cite{bias} (cf. \cite{hash, hash2, vazirani-phd}) to construct $\ell$-wise independent binary random variables with few random bits, instead of relying on Lemma~\ref{lem:hash}.
This technique provides an analogous bound on additive estimation error as in Lemma~\ref{lem:hash} for events that depend on at most $\ell$ variables.
A family of $N$ such variables can be constructed using $\Oh(\ell + \log\log N + \log(\frac 1 \delta))$ random bits~\cite[Lemma 4.2]{bias}.

Since we are interested in having $N = \Oh(p \cdot \log q)$ variables, $\delta = \frac{1}{2q}$, and $(3\log q)$-wise independency,
the size of the whole probabilistic space becomes $2^{\Oh(\log q + \log\log p)} = q^{\Oh(1)}(\log p)^{\Oh(1)}$.
The problem is that we need to optimize the exponent at $q$ in order to obtain better lower bounds.
Unfortunately we are not aware of any construction of a $\delta$-biased $\ell$-wise independent hash family, that would optimize this constant.

\maindmc*
\begin{proof}
Let us fix $\eps > 0$.
We are going to reduce the exact version of $p$-\dmc with 4 terminals, which is W[1]-hard, to the version with a~sufficiently large gap, parameterized by both $p$ and $k=|\mathcal{T}|$.
Let $L$ be an integer larger than~$\frac 2 \eps$.

For an instance $(G,\mathcal{T},p)$ of $p$-\dmc we apply Lemma~\ref{thm:dmc-gap} with $q = p^L$.
If~the original instance is fully solvable, the new one is as well.
Otherwise the maximal fraction of~separated terminal pairs is $\frac{k_0}{q} = \Oh(p\cdot q\log q) = \Oh(p^{L+2})$.
On the other hand, $\frac{k_0}{\alpha(k_0)} = \Omega\left(k_0^{(\frac 1 2 + \eps)}\right) \ge \Omega\left(p^{2L\cdot(\frac 1 2 + \eps)}\right)$.
The exponent at $p$ in the latter formula is $L + 2\eps L > L + 2$, so for large $p$ it holds $\frac{k_0}{\alpha(k_0)} \ge \frac{k_0}{q}$, therefore the reduction maps NO-instances into those where all cuts of size $p_0$ separate at most $\frac{k_0}{\alpha(k_0)}$ terminal pairs.
Both $k_0$~and~$p_0$ are functions of $p$,
therefore we have obtained a~parameterized reduction.
\end{proof}

\polydmc*
\begin{proof}
Suppose there is such an algorithm for some $\eps>0$ and proceed as in the proof of~Theorem~\ref{thm:dmc-w1} with
$L$ sufficiently large so that $2\eps L \ge 5$ and
 $q=m^L$.
The reduction is polynomial because $L$ is constant for fixed~$\eps$. 
We have $m_0 = |E(H)| = m\cdot p_0 + \Oh(k_0 \cdot p_0) = mp^2\log q + \Oh(p^3q^2\log^2 q) = \Oh(m^{2L+5})$ because 
$p \le m$.
If~the initial instance is fully satisfiable, then always $S(H,\mathcal{T}_H,p_0) = k_0$.
For a NO-instance, we have $S(H,\mathcal{T}_H,p_0) \le \frac {k_0}{q} = \Oh(m^{L+2})$ with probability at least $\frac 1 2$.
On the other hand, $k_0 = \Omega(m^{2L})$ and

$$\frac{k_0}{\alpha(m_0)} = \Omega\left(\frac{1}{m_0^{\frac 1 2 - \eps}}\right)\cdot k_0 = 
\Omega\left(m^{2L - (2L+5)\cdot (\frac 1 2 - \eps)}\right)  = \Omega(m^{L - \frac{5}{2}+ 2\eps L}).$$

We have adjusted $L$ to have $L - \frac{5}{2}+ 2\eps L > L + 2$, so
for large $m$ we get $\frac{k_0}{\alpha(m_0)} \ge \frac{k_0}{q}$.
Therefore the reduction maps NO-instances into those where all cuts of size $p_0$ separate at most $\frac{k_0}{\alpha(m_0)}$ terminal pairs.
When the reduction from Lemma~\ref{thm:dmc-gap} is correct (with probability at least  $\frac 1 2$), we are able to detect the NO-instances.
This implies that \dmc $\in$ co-RP.
\end{proof}

\section{W[1]-completeness of \so}
\label{sec:w1}

In this section we present a tight upper bound for the parameterized hardness level of $k$-\so, complementing the known $W[1]$-hardness.
We construct an FPT reduction to $k$-\textsc{Clique} and thus show that the problem belongs to $W[1]$.
The fact that $k$-\so is $W[1]$-complete implies the same for the gap version of the problem studied in the previous chapter,
what is uncommon is the theory of parameterized inapproximability.

The main idea in the reduction is to restrict to solutions
consisting of $f(k)$ subpaths, which can be chosen almost freely between fixed endpoints.
This formalizes an intuitive observation that different 
terminal pairs should obstruct each other only limited number of times,
because otherwise one path may exploit the other one as a~shortcut
and the whole knot of obstructions could be disentangled.


\begin{definition}[Canonical path family]
For a graph $G$, a family of paths $(P_{v,u})$, defined for all pairs $(v,u)$ such that $u$ is reachable from $v$ in $G$, is called canonical if it satisfies the following conditions:
\begin{enumerate}
    \item if there exists an edge $(u,v)$ in $G$, then $P_{u,v} = (u, v)$,
    \item if a vertex $w$ lies on a path $P_{u,v}$, then $P_{u,v}$ equals $P_{u,w}$ concatenated with $P_{w,v}$.
\end{enumerate}
\end{definition}

There might be multiple choices for such a family, but
for our purposes we just need to fix an arbitrary one.
It is important that we are able to construct it efficiently.

\begin{lemma}
\label{lem:canon-family}
A canonical path family can be constructed in polynomial time.
\end{lemma}
\begin{proof}
Let us fix any labeling of vertices with a linearly ordered alphabet.
We define $P_{u,v}$ to be the shortest path from $u$ to $v$, breaking ties lexicographically.
In order to construct it, we begin with computing distances between all pairs of vertices.
The first edge on the path $P_{u,v}$ goes to the lexicographically smallest vertex among those minimizing distance to $v$, and similarly for the further edges.

If there exists an edge $(u,v)$ in $G$, then it constitutes the only shortest path from $u$ to~$v$, therefore $P_{u,v} = (u, v)$.
To check property (2), suppose that $w$ lies on $P_{u,v}$ and let $P'$ be the subpath of $P_{u,v}$ from $u$ to $w$.
By definition, $P_{u,w}$ is the lexicographically smallest $(u,w)$-path among the shortest ones.
If $P' \ne P_{u,w}$, then we could replace $P'$ with $P_{u,w}$ and either shorten $P_{u,v}$ or decrease its position in the lexicographic order.
This is not possible, since $P_{u,v}$ is the lexicographically smallest $(u,v)$-path among the shortest ones.
Therefore $P' = P_{u,w}$ and analogous argument works for~$P_{w,v}$.
\end{proof}

\begin{definition}[Support]
Suppose a simple path $P = (v_1, \dots, v_m)$ can be represented as
a concatenation of canonical paths $P_{u_i, u_{i+1}}$ for some
sequence $v_1 = u_1, \dots, u_\ell = v_m$.
We will refer to the set $\{u_1, u_2, \dots, u_\ell\}$ as a support of $P$.
A path may have multiple supports.
If a path admits a~support of size at most $\ell$, then we say it is $\ell$-canonical. 
\end{definition}

{Thanks to property (1), the entire vertex set of a path forms a~support, so every path is $n$-canonical.
However, the same path might also admit a~significantly smaller support.
We are going to show that every solution can be transformed into one consisting of paths with moderate supports.
}


\begin{definition}[Schedule]
Consider a family of non-empty sets $A_1, A_2, \dots, A_k \subseteq [n]$.
For each $i \in [k]$ we consider variables $a^{start}_i = \min A_i$ and $a^{end}_i = \max A_i$.
The~schedule of this family is the order relation over those $2k$ variables.
\end{definition}

For example, if $A_1 = \{1,3,4\},\, A_2 = \{2,4\}$, then the schedule is given by relation $a^{start}_1 < a^{start}_2 < a^{end}_1 = a^{end}_2$.
It is easy to see that the number of possible schedules for $k$ sets is at most $(2k)^{2k}$: for each of the $2k$ variables we choose its position in the sequence, allowing different variables to share the same position.

We are ready to prove the main observation needed for the proof.
A cycle (resp. path) in a mixed graph is an~edge set that can be oriented to form a~directed cycle (resp. path).
A~mixed graph $G$ is acyclic, if any orientation $\og$ is an~acyclic directed graph.
In particular, any undirected connected component in an~acyclic mixed graph forms a~tree.
It is well known~\cite{xp} that $k$-\so on general mixed graphs reduces to the case of acyclic mixed graphs, since any cycle
can be contracted without changing the answer to the instance.
We say that two paths are in \emph{conflict} if there is an undirected edge they use in different directions.

\begin{lemma}\label{lem:canonical}
If a $k$-\so instance is satisfiable on an acyclic mixed graph, then it admits a solution in which each path is $k^{\Oh(k)}$-canonical.
\end{lemma}
\begin{proof}
A solution is a family of paths $P_1, \dots, P_j$ which are not in conflict and $P_j$ is a~ $(s_j,t_j)$-path.
We say that a set $V' \subseteq V$ is a support of the solution if
it contains supports for all $P_j$.
Due to property $(1)$ of Definition~\ref{lem:canon-family}, every path admits some, potentially large, support,
so we can always find such a set $V'$.
We are going to show that if an instance is solvable, then there exists
a solution with support of size
less than $\beta_k = 4k\cdot 4^k \cdot (2k)^{2k}$.
To do so, we assume the contrary: that the minimal size of a support is at least $\beta_k$, and then construct a solution with a smaller support.
This will entail the claim.

Consider a solution with support $V'$, $|V'| \ge \beta_k$.
For a vertex $u \in V'$, we define $R(u)$ to be the set
of indices $j$, such that $P_j$ goes though $u$.
By counting argument, there must be at least $4k\cdot 2^k \cdot (2k)^{2k}$ vertices in $u_i \in V'$ with the same non-empty set $R = R(u_i)$.

There is a natural linear ordering $u_1,u_2\dots,u_\ell$ of those vertices, that is coherent with the orientation of paths from the terminal set $R$. 
For each even index $i$, consider the canonical path from $u_{i-1}$ to $u_{i+1}$ and call it $Q_i$.
If such a path is not in conflict with paths from $[k] \setminus R$,
we could remove $u_i$ from $V'$ -- all paths from $R$ can have $u_i$ removed from their support, and other paths do not go through $u_i$.
This would contradict $V'$ being the minimum size support.

For the canonical path $Q_i$ 
we first consider which other paths are in conflict with $Q_i$: let us call this set $T(i) \subseteq [k]$.
Next, we look at the schedule of vertex sets $Q_i \cap P_j$ for $j \in T(i)$ with respect to the order given by $Q_i$.
Again by counting argument, we can pick $2k$ paths $Q_{j_1}, \dots, Q_{j_{2k}}$ with the same $T = T(j_i)$ and the same schedule
(factor 2 for choosing only even indices, at most $2^k$ choices of $T$, at most $(2k)^{2k}$ different schedules).
Recall that the first vertex of $Q_{j_{i+1}}$ is reachable from the last vertex of $Q_{j_i}$.

\begin{figure}
\centering
\includegraphics[scale=0.85]{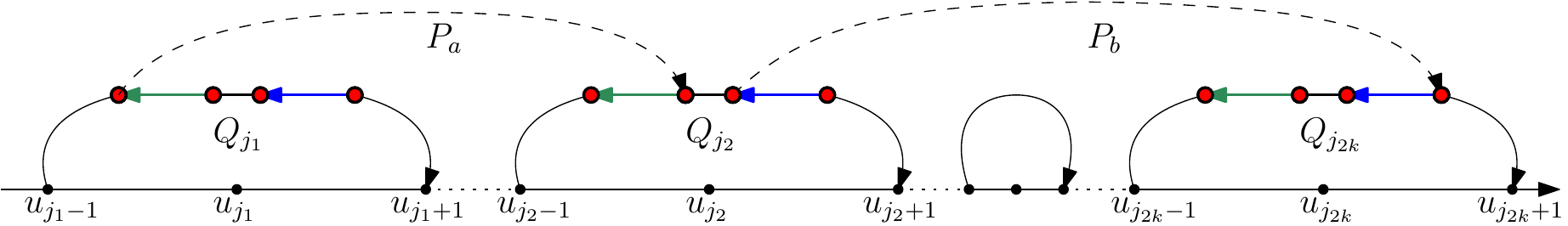}
\caption{ The construction of a detour. The red nodes are vertices
where the path switches between canonical $Q_{j_1}$ and subpaths of $P_a$ (the green one) or $P_b$ (the blue one).
They always appear in the same order because we have chosen $Q_{j_i}$ to share the same schedule.
Note that the last subpath goes directly to $Q_{j_{2k}}$, potentially omitting, e.g., $Q_{j_{2k - 1}}$.
}\label{fig:w1}
\end{figure}

Having the fixed schedule makes it possible to find a detour from $u_{j_{1}-1}$ (the beginning of $Q_{j_1}$) to $u_{j_{2k}+1}$ (the end of $Q_{j_{2k}}$) that omits the vertices $u_{j_i}$ in between.
Let us start by following $Q_{j_1}$ to the place of~the first conflict with some $P_a$.
We know that there is a path from the first vertex of $Q_{j_1} \cap P_a$ to the last vertex of $Q_{j_2} \cap P_a$, which is a subpath of $P_a$
($P_a$ cannot go in the other direction because we assumed the graph to be acyclic). 
We can follow this path and continue on $Q_{j_2}$ to the next occurrence of some conflict with $P_b$, then use the same argument to reach the last vertex of $Q_{j_3} \cap P_b$ and iterate this procedure.
Note that the vertex we 'arrive' at in $Q_{j_i}$ might be the same where the next detour starts.
It is crucial that we have fixed a single schedule so we will never reach another conflict with the same path $P_a$.
We need at most $|T| \le k-1$ iterations to get behind all conflicts -- then the last detour is guaranteed to reach a vertex in $Q_{j_{2k}}$ (potentially omitting several segments $Q_{j_i}$ with a~single subpath of $P_a$), from where we can reach the end of $Q_{j_{2k}}$, that is,  $u_{j_{2k}+1}$.
The idea of a detour is depicted
in Figure~\ref{fig:w1}.

We can now remove all the vertices $u_{j_1}, \dots, u_{j_{2k}}$ from $V'$
and replace the $(u_{j_1 - 1}, u_{j_{2k} + 1})$-subpath with the detour constructed above, for all terminal pairs in $R$.
Thanks to property $(2)$ of Definition~\ref{lem:canon-family},
we might need to add at most $2 k-2$ vertices to $V'$ (2 for each iteration) and the rest of the detour is either canonical or follows paths that are already in the solution.
Therefore we have constructed a~solution with  $V'$ smaller by at least $2k - (2k-2) = 2$, which finishes the proof.

\end{proof}

\mainhard*
\begin{proof}
The problem is known to be W[1]-hard~\cite{pilipczuk-multicut}, so we just need
 a parameterized reduction to $k$-\textsc{Clique}.
Let $(G,\mathcal{T})$ be an instance of $k$-\so and let $\beta_k = k^{\Oh(k)}$ be the sequence from Lemma~\ref{lem:canonical}.
As already mentioned, we can assume $G$ to be acyclic~\cite{xp}.
We construct the instance~$H$ of $(k\beta_k)$-\textsc{Clique} as follows:
\begin{enumerate}
\item Compute a canonical family of paths for $G$ in polynomial time (Lemma~\ref{lem:canon-family}).
\item For each pair $(i,j) \in [k] \times [\beta_k]$ create an independent set $H_{i,j}$.
The vertices in $H_{i,j}$ are given as ordered pairs $(u,v)$, such that $v$ is reachable from $u$ in $G$.
If $j=1$, then we require $u = s_i$ and if $j=\beta_k$, we require $v = t_i$.
We allow pairs of form $(u,u)$.
\item For vertices $(u_1, v_1)$ and $(u_2, v_2)$ lying in distinct sets $H_{i,j}$, place an edge between them if~the canonical paths $ P_{u_1, v_1}$ and $ P_{u_2, v_2}$ are not in conflict.
\item If we placed an edge between $(u_1, v_1) \in H_{i,j}$ and $(u_2, v_2) \in H_{i,j+1}$, remove it unless $v_1 = u_2$.
\end{enumerate}

The size of $V(H)$ is bounded by $k\beta_k\cdot |V(G)|^2$.
If the constructed graph $H$ admits a~clique of~size $k\beta_k$,
then each of its vertices must lie in a different independent set $H_{i,j}$.
Due to step (4), we know that canonical paths encoded by the choice of vertices in $H_{i,1}, H_{i,2}, \dots, H_{i,\beta_k}$ match and they form a path from $s_i$ to $t_i$.
Step (3) ensures that those paths are not in conflict, therefore the instance $(G,\mathcal{T})$ is satisfiable.

On the other hand, if $(G,\mathcal{T})$ admits a solution, we can assume its paths to be $\beta_k$-canonical due to Lemma~\ref{lem:canonical}.
We can thus choose the vertices in $H$ to reflect their supports,
padding them with trivial paths $(t_i,t_i)$ if the support is smaller than $\beta_k$.
Since none of the paths are in~conflict, there is an edge in $H$ between all chosen vertices.
\end{proof}

\section{Final remarks and open problems}

I would like to thank Pasin Manurangsi for helpful discussions and, in particular, for suggesting the argument based on Chernoff bound {in Lemma~\ref{thm:so-gap}},
which is surprisingly simple and powerful.
A question arises whether one can derandomize this argument efficiently and construct a $\delta$-biased sampler family in an~explicit way.
This would allow us to replace the assumption NP $\not\subseteq$ co-RP with P $\ne$ NP for \dmc.
This technique may also find use in other reductions in parameterized inapproximability.

An obvious question is if any of the studied problems admits an $o(k)$-approximation, or~if~the lower bounds can be strengthened.
Note that for the maximization version of \dmc we do not know anything better than $\frac k 2$-approximation as we cannot solve the exact problem for $k>2$.
For \so, the reason why the value of the parameter in the self-reduction becomes so large, is~that in each step we can add only an~exponentially small term to the gap.
Getting around this obstacle should lead to better lower bounds.
Also, the approximation status  for $k$-\so on planar graphs remains unclear~\cite{planar-hardness}.
Here we still cannot rule out a constant approximation and there are no upper bounds known.

Finally, it is an open quest to establish relations between other hard parameterized problems and their gap versions.
Is $F(k)$-\textsc{Gap $k$-Clique} W[1]-hard for $F(k) = o(k)$ and is $F(k)$-\textsc{Gap $k$-Dominating Set} W[2]-hard for any function $F$ (open questions in~\cite{dominating-set})?
Or can it be possible that $F(k)$-\textsc{Gap $k$-Dominating Set} is in W[1] for some function $F$?
\bibliographystyle{plain}
\bibliography{steiner}

\end{document}